\documentclass[journal]{IEEEtran}
\usepackage{cite}
\usepackage{hyperref}  % For making links clickable

\ifCLASSINFOpdf
  \usepackage[pdftex]{graphicx}
\else
\fi
\usepackage{amsmath}
\usepackage{amssymb}   % For \mathbb
\usepackage{amsthm}    % For \theoremstyle

\usepackage{xcolor}  % For colors

\usepackage{array}

\def\ds{\displaystyle}
\DeclareMathOperator{\esssup}{ess\,sup}
\newtheorem{algorithm}{Algorithm}
\newtheorem{assumption}{Assumption}

\newtheorem{theorem}{Theorem}
\theoremstyle{definition}\newtheorem{definition}{Definition}
\theoremstyle{remark}\newtheorem{remark}{Remark}
\theoremstyle{remark}

\onecolumn

\begin{document}
%
% paper title
% Titles are generally capitalized except for words such as a, an, and, as,
% at, but, by, for, in, nor, of, on, or, the, to and up, which are usually
% not capitalized unless they are the first or last word of the title.
% Linebreaks \\ can be used within to get better formatting as desired.
% Do not put math or special symbols in the title.
\title{Robust Quickest Change Detection with Cost-Efficient Experiment Design and Unknown Post-Change Model}
%
%
% author names and IEEE memberships
% note positions of commas and nonbreaking spaces ( ~ ) LaTeX will not break
% a structure at a ~ so this keeps an author's name from being broken across
% two lines.
% use \thanks{} to gain access to the first footnote area
% a separate \thanks must be used for each paragraph as LaTeX2e's \thanks
% was not built to handle multiple paragraphs
%

\author{Patrick Vincent N. Lubenia
        and Taposh Banerjee% <-this % stops a space
\thanks{This work was supported in part by the U.S. National Science Foundation under Grant 2427909. (\textit{Corresponding author: Taposh Banerjee.})}% <-this % stops a space
\thanks{The authors are with the University of Pittsburgh, Pittsburgh, PA 15260 USA (email: pnl8$@$pitt.edu, taposh.banerjee$@$pitt.edu).}% <-this % stops a space
% \thanks{J. Doe and J. Doe are with Anonymous University.}% <-this % stops a space
% \thanks{Manuscript received April 19, 2005; revised August 26, 2015.}
}

\maketitle

% As a general rule, do not put math, special symbols or citations
% in the abstract or keywords.
\begin{abstract}
A quickest change detection problem is studied where the decision-maker has access to multiple experiments, but can choose only one experiment at a time to detect the change. The true post-change distribution under each experiment is unknown, but known to belong to a family of distributions. 
It is assumed that each of these families of distribution have an element that is least favorable in a well-defined sense. 
An algorithm is proposed that is designed using the least favorable distributions. It is shown that the algorithm is first-order asymptotically optimal, as the rate of false alarms goes to zero. 
This algorithm minimizes the worst-case average detection delay over all post-change distributions within the assumed families, subject to a constraint on the average run length to false alarm and constraints on the average fraction of time each experiment is performed before change. 
\end{abstract}

% Note that keywords are not normally used for peerreview papers.

% For peer review papers, you can put extra information on the cover
% page as needed:
% \ifCLASSOPTIONpeerreview
% \begin{center} \bfseries EDICS Category: 3-BBND \end{center}
% \fi
%
% For peerreview papers, this IEEEtran command inserts a page break and
% creates the second title. It will be ignored for other modes.
\IEEEpeerreviewmaketitle

% ------------ %
% Introduction %
% ------------ %

\section{Introduction}
% The very first letter is a 2 line initial drop letter followed
% by the rest of the first word in caps.
% 
% form to use if the first word consists of a single letter:
% \IEEEPARstart{A}{demo} file is ....
% 
% form to use if you need the single drop letter followed by
% normal text (unknown if ever used by the IEEE):
% \IEEEPARstart{A}{}demo file is ....
% 
% Some journals put the first two words in caps:
% \IEEEPARstart{T}{his demo} file is ....
% 
% Here we have the typical use of a "T" for an initial drop letter
% and "HIS" in caps to complete the first word.
\IEEEPARstart{I}{n} the classical quickest change detection (QCD) problem, a decision-maker observes a stochastic process whose distribution changes at a fixed unknown point, called the change-point. 
The observations are collected by performing a single experiment repeatedly over time. The goal is to detect this change with the minimum possible delay, subject to a constraint on the rate of false alarms. The QCD problem has applications in statistical process control, sensor network surveillance, computer network security, network analysis, finance, and public health, among others  \cite{poor2009, tartakovsky2014, tartakovsky2019, veeravalli2013}.

% observe the process with independent and identically distributed observations. Furthermore, the pre-change and post-change distributions are known to the observer. 

When the change point is modeled as a random variable, the QCD problem is said to be in the Bayesian setting. Some results on Bayesian QCD can be found in \cite{veeravalli2013, tartakovsky2014, tartakovsky2019, banerjee2021, guo2023, naha2024}.
When the change-point is modeled as a
deterministic constant, the QCD problem is said to be in non-Bayesian settings. In this domain, a minimax formulation is often preferred. 
The classical QCD problem has been studied in two minimax formulations in \cite{lorden1971} and \cite{pollak1985}. An optimal solution for both formulations is the cumulative sum (CUSUM) algorithm developed in \cite{page1954}. It was shown in \cite{lorden1971} that this algorithm is asymptotically optimal as the false alarm rate goes to zero. Subsequently, the CUSUM algorithm was proven to be exactly optimal for any false alarm rate, as detailed in \cite{moustakides1986}. A comprehensive review of non-Bayesian QCD can be found in \cite{veeravalli2013, tartakovsky2014, tartakovsky2019, liang2023, brucks2023}.

In practical applications, performing an experiment for data collection typically has an associated cost, e.g., energy consumption, labor costs, and lab tests. When change occurs rarely, these costs may become economically burdensome. This issue has motivated the study of sampling control in the statistical process control literature \cite{reynolds2004, tagaras1998survey}. In the context of QCD, the issue of the cost of observations has been investigated 
in multiple studies, such as in \cite{banerjee2012, banerjee2013, banerjee2015, banerjee2015-2, banerjee2015-3}. The optimal algorithms in \cite{banerjee2012, banerjee2013, banerjee2015, banerjee2015-2, banerjee2015-3} minimize detection delay while controlling the false alarm rate and the average number of observations used or the fraction of time observations are used prior to the change-point.
These works on on-off observation control have been partially extended to a QCD problem with two experiments in \cite{ren2017}. 

In recent work, we developed a theory of QCD with cost-efficient experiment design involving multiple experiments \cite{lubenia2026}. In the problem in \cite{lubenia2026}, the decision-maker has access to multiple experiments, but it can choose only one experiment at a time. It is assumed that the experiments could be ordered according to their quality, as measured by their Kullback-Leibler (KL) divergence, a measure of how close two distributions are to each other. The best experiment is known to the decision-maker, but it cannot perform it all the time as it has to satisfy a constraint on the cost of experiment design. Specifically, the goal in  \cite{lubenia2026} is to find a policy that minimizes the worst-case average detection delay subject to constraints on the average run length to a false alarm and the pre-change observation ratio, the fraction of time an experiment is performed before change. The formulation allows data efficiency, where the decision-maker has the option not to perform any experiment. To solve the problem, we proposed the multiple-experiment CUSUM ($m$E-CUSUM) algorithm, a generalization of the CUSUM algorithm to multiple experiments. This algorithm is asymptotically optimal, as the rate of false alarms goes to zero, and it can achieve the performance of the best experiment. Using techniques of scaling of undershoots, soft truncation of tests, and randomization, we showed that the cost constraints can be met exactly. This is unlike other efforts in the literature where the constraints can often only be met with 
strict inequality \cite{ren2017}. We have also introduced an iterative technique to design the algorithm: the algorithm for the two-experiment case (when there are only two experiments to choose from) is proposed, and then the algorithm for the three-experiment case is defined in terms of the algorithm for the two-experiment case, and so on. 
The area of sequential experiment design is vast, and we refer to Section I-A of \cite{lubenia2026} for a detailed literature survey that provides further justification for our work.

In this paper, we extend the multiple-experiment method developed in \cite{lubenia2026} to the case where the post-change distribution is unknown. In the literature on QCD with unknown post-change distributions, three main approaches are typically used to find optimal algorithms. First, the generalized likelihood ratio tests replace the unknown post-change parameter with its maximum likelihood estimate \cite{tartakovsky2014, tartakovsky2019, lorden1971, lai1998, lau2019}. Second, mixture-based tests assume a prior for the post-change parameter, and the likelihood ratio is integrated over this prior \cite{tartakovsky2014, tartakovsky2019, lai1998, pollak1987}. Third, robust tests consider a least favorable law (LFL) for the post-change distribution \cite{hou2024, hou2025, unnikrishnan2011, oleyaeimotlagh2023}. Among the three approaches, the robust approach often leads to computationally efficient algorithms whose statistics can be computed recursively. Motivated by the latter, we take a robust approach in this paper. 

In our problem setup, we assume that the post-change family of laws under each experiment has a known law that is least favorable in a well-defined sense. Our solution approach is to first order the experiments according to the information quality of their least favorable laws (LFLs). Thus, the experiment for which the KL-divergence between its post-change LFL and its pre-change law is the highest is considered to be the best experiment. We then use the mE-CUSUM algorithm from \cite{lubenia2026} to perform QCD and cost-efficient experiment design using the LFLs. We call this LFL-based algorithm the multi-Experiment Robust CUSUM (mE-RCUSUM) algorithm. We show that the mE-RCUSUM algorithm is asymptotically optimal with respect to a minimax robust formulation that we propose in this paper. The problem of robust QCD with on-off observation control has been studied in \cite{hou2025}. In this paper, we extend that study to multiple experiments.

We note that the order of the quality of the experiments based on their true KL-divergences (which is unknown, since the true post-change parameters are unknown) may be different from the order based on the quality of the LFLs. We elaborate on this observation using an example. Consider an intruder detection system which has an audio sensor and a video sensor. The manifestation of an intrusion in each sensor (volume of footsteps or visibility of an intruder) is unknown in advance, and the observer cannot assume a fixed post-change signal for each sensor. Instead, the decision-maker can plan for the worst-case scenario: the least favorable or hardest-to-detect intrusion signature each sensor could produce. Due to deteriorating quality, the video sensor's worst-case signal is almost indistinguishable from ordinary background clutter, while the audio sensor's worst-case signal remains comparatively more separable from ambient noise. This worst-case separability defines quality in our framework: the decision-maker considers the audio sensor to be of higher quality than the video sensor. Due to economic considerations, an observer cannot use both sensors all the time. At each observation time, the decision-maker can choose to 1) observe the audio sensor, 2) observe the video sensor, 3) observe both at the same time, or 4) not observe any sensor at all. %Our study generalizes this to $m$ experiments, each with an unknown post-change distribution and corresponding least favorable law.

The rest of the paper are organized as follows. Section~\ref{sec:prob} provides details of the problem formulation for the robust multiple-experiment QCD problem with cost constraints. We then present the $m$E-RCUSUM algorithm in Section~\ref{sec:me-rcusum}. To appreciate the algorithm, Section~\ref{sec:example} provides an example illustrating the algorithm. Theoretical guarantees of the $m$E-RCUSUM algorithm are discussed in Section~\ref{sec:optimality} before summarizing the study in Section~\ref{sec:conclusion}.

\section{Problem Formulation}
\label{sec:prob}

In this section, we formulate the problem of quickest change detection (QCD) with cost-constrained experiment design where the post-change distribution is unknown under each experiment, but known to belong to a family of distributions.

Consider a stochastic process whose statistical properties change at time $\nu$. At each time point, a decision-maker observes the process by performing one of $m$ experiments. Let $X_n^i$ be the observation at time $n$ using experiment $i$. The sequence of random variables $\{ X_n^i \}$ has a probability density function $f_i$ for $n < \nu$ and $g_i$ for $n \geq \nu$. The densities $\{f_i\}_{i=1}^m$ are known but the densities $\{g_i\}_{i=1}^m$ are unknown. However, we assume that, for each $i$, $g_i$ belongs to a known family of distributions $\mathcal{G}_i$. We further assume that each family has a least favorable law (LFL) $\bar{g}_i \in \mathcal{G}_i$. Definition~\ref{def:stoch-bdd} below defines the notion of stochastic boundedness, which will allow us to formally define the notion of LFL in Definition~\ref{def:stoch-bdd-law}.

\begin{definition}
\label{def:stoch-bdd}

Consider random variables $Z_1$ and $Z_2$ with density $f_1$ and $f_2$, respectively. We say that $f_2$ (or $Z_2$) is stochastically bounded below by $f_1$ (or $Z_1)$, denoted $f_1 \prec f_2$ (or $Z_1 \prec Z_2$), if
$$ \mathsf{P}(Z_1 \geq t) \leq \mathsf{P}(Z_2 \geq t), \ \forall t \in \mathbb{R}. $$
\end{definition}

\begin{definition}
\label{def:stoch-bdd-law}

For each $i$, we say that the family of post-change laws $\mathcal{G}_i$ has a least favorable law (LFL) $\bar{G}_i$ with density $\bar{g}_i$ if
$$ \mathcal{L}(\phi (X), \bar{G}_i) \prec \mathcal{L}(\phi (X), G_i), \ \forall G_i \in \mathcal{G}_i, $$
where $\mathcal{L}(\phi (X), G_i)$ denotes the law or distribution of $\phi (X)$ when $X$ has distribution $G_i$.
\end{definition}

We assume that the sequence of experiments is ordered in the following manner: 
%The sequence $\{ X_n^j \}$ has a higher information quality compared to $\{ X_n^i \}$ for $i < j$, as measured by their Kullback-Leibler (KL) divergence:
\begin{equation}
0 < D(\bar{g}_1 \ \Vert \ f_1) \leq D(\bar{g}_2 \ \Vert \ f_2) \leq \dots \leq D(\bar{g}_m \ \Vert \ f_m) < \infty, \label{eq:kl}
\end{equation}
i.e., the LFL of experiment $m$ has the highest information quality, the LFL of experiment $m - 1$ the next-highest, and so on, and the LFL of experiment 1 has the lowest information quality. This ordering implies that the worst-case signal of experiment $i$ is more informative than the worst-case signal of experiment $i - 1$. We note that the actual quality of the experiment is unknown since the post-change distribution under each experiment is unknown. Thus, the true order of KL-divergences of the $m$ experiments may be different from the order specified in \eqref{eq:kl}. For example, it is possible that the true KL-divergence for experiment $m$ is the smallest among all experiments. But the uncertainty set for the post-change law for this experiment is much smaller as compared to that of others, making the KL-divergence of its LFL the largest. For robust QCD with experiment design, it is the ordering specified in \eqref{eq:kl} that matters for optimality. 

At each time $n \geq 0$, the decision-maker decides whether or not a change has occurred based on the information available at time $n$. If the available information suggests no change has occurred, the decision-maker decides which experiment to perform at time $n + 1$. To formalize this, let the decision variable $S_n^i$ and decision vector $S_n$ be defined as follows:
\begin{align*}
S_n^i & = \left\{
\begin{array}{cl}
1 & \text{if experiment $i$ is performed}, \\
0 & \text{otherwise},
\end{array}
\right. \\
S_n & = (S_n^1, \dots, S_n^m).
\end{align*}
Moreover, since at most one experiment can be performed at each time $n$, then
$$ \sum_{i = 1}^m S_n^i \leq 1. $$
When this sum is equal to 0, then no experiment is performed.

$S_n$ is a function $\phi_{n - 1}$ of the information $I_{n - 1}$ available up to time $n - 1$:
$$ S_n = \phi_{n - 1} (I_{n - 1}) \ \text{ for } n = 1, 2, \dots, $$
where $\phi_{n - 1}$ denotes the control law at time $n - 1$. Here, we define the information available at time $n$ using the vectors
\begin{align*}
I_n^i & = \left[ S_1^i, \dots, S_n^i, (X_1^i)^{(S_1^i)}, \dots, (X_n^i)^{(S_n^i)} \right] \\
I_n & = (I_n^1, \dots, I_n^m)
\end{align*}
where $I_0^i = \varnothing$ for $i = 1, \dots, m$. For $k = 1, \dots, n$, $(X_k^i)^{(S_k^i)}$ represents $X_k^i$ if $S_k^i = 1$. Otherwise, $X_k^i$ is not in $I_n^i$.

At any time $n$, using the information $I_n$ available at that time, the decision-maker decides whether a change has occurred using a stopping time $\tau$. Therefore, a policy for the multiple-experiment QCD is $\Psi = \{ \tau, \phi_0, \dots, \phi_{\tau - 1} \}$.

To determine an optimal policy, we define metrics we will use for the delay, false alarm, and observation costs. We define first the set $\mathcal{G}$ of all possible environments, where an environment is a specification of every experiment's post-change law:
$$ \mathcal{G} = \mathcal{G}_1 \times \mathcal{G}_2 \times \dots \times \mathcal{G}_m. $$
An environment $G \in \mathcal{G}$ is a specification of each experiment's post-change law simultaneously, even though the decision-maker only observes data from one experiment at a time.

For the detection delay, we consider Lorden's formulation for the worst-case average detection delay (WADD) for an environment $G \in \mathcal{G}$ \cite{lorden1971}:
$$ \mathsf{WADD}_G (\Psi) = \sup_{\nu \geq 1} \esssup \mathsf{E}_\nu^G [(\tau - \nu + 1)^+ \mid I_{\nu - 1}]. $$
Here, $\mathsf{E}_n^G$ is the corresponding expectation of the probability measure $\mathsf{P}_n^G$ which depends on which experiment is used at time $n$. Following the notations used in \cite{hou2025}, when no change occurs, i.e., $\nu = \infty$, there is no post-change law to speak of, and we use $\mathsf{P}_\infty = \mathsf{P}_\infty^G$ and $\mathsf{E}_\infty = \mathsf{E}_\infty^G$ to denote the corresponding probability measure and expectation, respectively.

For the false alarm metric, we consider the average run length to false alarm (ARLFA) which is used in \cite{lorden1971} and \cite{pollak1985}. We constrain it to be a nonnegative number $\gamma$:
$$ \text{ARLFA} (\Psi) = \mathsf{E}_\infty [\tau] \geq \gamma. $$
Finally, for the metric for the observation costs, we note that each experiment has a different cost. Thus, we seek control policies that meet a constraint on the fraction of time each experiment is performed before the change-point. We use the pre-change observation ratio (POR) introduced in \cite{lubenia2026} where $\text{POR}_i$ represents the POR of experiment $i$:
\begin{align*}
\mathsf{POR}_i (\Psi)
& = \limsup_{n \to \infty} \ \mathsf{E}_n \left[ \frac{1}{n} \sum_{k = 1}^{n - 1} S_k^i\right] \\
& = \limsup_{n \to \infty} \ \mathsf{E}_\infty \left[ \frac{1}{n} \sum_{k = 1}^{n - 1} S_k^i \right].
\end{align*}
From the above definition, it is clear that $0 \leq \text{POR}_i (\Psi) \leq 1$ for all $i$. If experiment $i$ is performed all the time in a policy, then $\text{POR}_i (\Psi) = 1$.
In addition, we note that since at most one experiment is performed at any given time, 
$$ 0 \leq \sum_{i = 1}^m \text{POR}_i (\Psi) \leq 1. $$
When the sum in the last inequality is strictly less than $1$, then there is a positive fraction of time when no experiment is performed. Note that since the ARLFA and POR values do not depend on the post-change enviroment $G$. 

Following the formulation in \cite{hou2025}, our robust optimization problem for the multiple-experiment QCD problem where the post-change law is unknown is given by
\begin{align}
\label{eq:prob-formulation}
\begin{split}
\inf_{\Psi} \ \sup_{G \in \mathcal{G}} \quad & \mathsf{WADD}_G (\Psi) \\
\text{s.t.} \quad & \mathsf{ARLFA} (\Psi) \geq \gamma \\
& \mathsf{POR}_i (\Psi) \leq \beta_i \text{ for } i = 1, 2, \dots, m,
\end{split}
\end{align}
where $\gamma > 0$, $0 < \beta_i < 1$, for $i = 1, 2, \dots, m$, are given constraints, and $\ds \sum_{i = 1}^m \beta_i < 1$. The gap $\ds \beta_0 = 1 - \sum_{i = 1}^m \beta_i$ is the constraint on the fraction of time no experiment is performed.

\section{Robust QCD Algorithms with Cost-Constrained Experiment Design}
\label{sec:me-rcusum}

In this section, we propose the \texorpdfstring{$m$}{}E-RCUSUM algorithm that solves the multiple-experiment problem defined in \eqref{eq:prob-formulation}. 
%We combine the LFL-based log-likelihood ratio used in \cite{hou2025} and the recursive structure of switching to multiple experiments utilized in \cite{lubenia2026}.
Before describing our algorithm, we make the following assumptions for the rest of the paper:
\begin{assumption}
Each post-change family of distributions $\mathcal{G}_i$ has an LFL $\bar{g}_i$ for $i = 1, 2, \dots, m$.
\end{assumption}
\begin{assumption}
    All likelihood ratios appearing in this paper are continuous and have finite positive moments up to second order.
\end{assumption}

% \begin{itemize}
% \item Assumption~1: Each post-change family of distributions $\mathcal{G}_i$ has an LFL $\bar{g}_i$ for $i = 1, 2, \dots, m$.

% \item Assumption~2: All likelihood ratios appearing in this paper are continuous and have finite positive moments up to second order.
% \end{itemize}

Our algorithm has the same structure as the algorithm in \cite{lubenia2026}, with the difference that here, we use the likelihood ratios calculated using the LFLs. As in \cite{lubenia2026}, we define the algorithm iteratively beginning with the case of two experiments ($m = 2$):
\begin{itemize}
    \item We first define the 2E-RCUSUM algorithm when we only have two experiments to choose from. See Algorithm \ref{alg:2e-rcusum} in Section~\ref{sec:2ERCUSUM}.
    \item Next, we define the Truncated 2E-RCUSUM algorithm, which is an important intermediate step in the design with multiple experiments. See Algorithm~\ref{alg:trunc-2e-rcusum} in Section~\ref{sec:Trunc_2ECUSUM}.
    \item Finally, we define the $m$E-RCUSUM algorithm for $m$ experiments in an iterative manner using the Truncated 2E-RCUSUM algorithm. See Algorithm~\ref{alg:me-rcusum} in Section~\ref{sec:mECUSUMAlgo}. 
\end{itemize}
The case where we only have one experiment ($m=1$) and must choose on-off observation control to save the cost of performing that experiment has been studied in \cite{hou2025}. Thus, in what follows, we only discuss the cases where 
$$
m \geq 2.
$$
In Section~\ref{sec:example}, we provide an example using Gaussian data that clarifies the structure of these algorithms. 

\subsection{The 2E-RCUSUM Algorithm for the case \texorpdfstring{$m=2$}{}}
\label{sec:2ERCUSUM}
In the case of two experiments ($m=2$), we denote the observations from the higher-quality experiment by $\{Y_n\}$ and those from the lower-quality experiment by $\{X_n\}$. The densities for experiment $Y$ are denoted by $f_Y$ and $g_Y$, and the densities for experiment $X$ are denoted by $f_X$ and $g_X$. The LFL for $Y$ is denoted by $\bar g_Y$, and the LFL for $X$ is $\bar g_X$.  

Using these notations, we propose the 
2E-RCUSUM Algorithm below. 

\begin{algorithm}[2E-RCUSUM: $\Psi_{2\text{RC}} (A, a_X, a_Y, N_0, N_X, \mu) \allowbreak = \Psi_{2\text{RC}}$]
\label{alg:2e-rcusum}

Fix $A > 0$, $a_X > 0$, $a_Y > 0$, $N_0 > 0$, $N_X > 0$, and $\mu > 0$. The \textup{2E-RCUSUM} statistic $\{ D_n \}$ is defined as follows:

\begin{enumerate}
\item We start with $D_0 = 0$ and update $\{ D_n \}$ using the observations $\{ Y_n \}$ from the higher-quality experiment using
$$ D_{n + 1}  = D_n + \bar{\ell}_Y (Y_{n + 1}), $$
where $\bar{\ell}_Y (y) = \log \frac{\bar{g}_Y (y)}{f_Y (y)}$, until 
$$ \sigma = \sigma^Y (A) = \min \{ n \geq 1: D_n \notin [0, A) \}. $$
If $D_{\sigma} \geq A$, we stop and declare the change.

\item If $D_{\sigma} < 0$, we start using observations $\{ X_n \}$ from the lower-quality experiment.

    \begin{enumerate}
    \item We first scale the undershoot $U_Y = D_{\sigma}$ by a factor $a_Y$.
    
    \item We update the statistic $\{ D_n \}$ starting at time $\sigma$ using
    $$ D_{n + 1} = D_n + \bar{\ell}_X (X_{n + 1}), $$
    where $\bar{\ell}_X (X_{n + 1}) = \log \frac{\bar{g}_X (x)}{f_X (x)}$, until either the statistic falls below $a_Y U_Y$ at time
    $$ \sigma^X (a_Y U_Y) = \min \{ n \geq \sigma + 1: D_n \not \in [a_Y U_Y, 0) \} $$
    or $N_X$ number of observations of experiment $X$ are used. If $D_{\sigma^X (a_Y U_Y)} > 0$, we reset the statistic $D_n$ to $0$ and go to Step $1$.
    
    \item If $D_{\sigma^X (a_Y U_Y)} < a_Y U_Y$, we do not perform any experiment but still scale the undershoot $U_X~=~D_{\sigma^X (a_Y U_Y)}$ by a factor of $a_X$ and use $a_X U_X$ as the starting level for deterministically updating $D_n$ as
    $$ D_{n + 1} = D_n + \mu. $$
    This continues until either $D_n$ goes above $a_Y U_Y$ or $N_0$ number of deterministic updates are used. At this time, we reset $D_n$ to $a_Y U_Y$ and continue updating $D_n$ as in Step $\textup{c}$.
    \end{enumerate}

\item Stop at the stopping time
\begin{align*}
\tau_{\scriptscriptstyle \textup{2RC}}
& = \tau_{\scriptscriptstyle \textup{2RC}} (A, a_X, a_Y, N_0, N_X, \mu) \\
& = \inf \{ n \geq 1: D_n \geq A \}.
\end{align*}
\end{enumerate}    
\end{algorithm}

\medskip
\subsection{The Truncated 2E-RCUSUM Algorithm}
\label{sec:Trunc_2ECUSUM}
We wish to define the algorithm for $m$ experiments using the algorithm for $m - 1$ experiments. To define the algorithm in an iterative manner, we use a notion of a truncated test. In the 2E-RCUSUM algorithm, there was no limit in the number of consecutive observations of experiment $Y$, but there was for experiment $X$. We define the truncated 2E-RCUSUM algorithm as the 2E-RCUSUM algorithm wherein we limit the number of consecutive observations of experiment $Y$. While performing experiment $Y$, if the statistic crosses the detection threshold or $N_Y$ observations are reached, we stop and declare a change. While the details are intuitively clear from the definition in Algorithm~\ref{alg:2e-rcusum}, we provide the full description of the truncated algorithm here for completeness. 

\begin{algorithm}[Truncated 2E-RCUSUM: $\Psi_{\text{T2RC}} (A, a_X, a_Y, \allowbreak N_0, N_X, N_Y, \mu) = \Psi_{\text{T2RC}}$]
\label{alg:trunc-2e-rcusum}

Fix $A > 0$, $a_X > 0$, $a_Y > 0$, $N_0 > 0$, $N_X > 0$, $N_Y > 0$, and $\mu > 0$. The truncated \textup{2E-RCUSUM} statistic $\{ D_n \}$ is defined as follows:

\begin{enumerate}
\item We start with $D_0 = 0$ and update $\{ D_n \}$ using the observations $\{ Y_n \}$ from the higher-quality experiment using
$$ D_{n + 1}  = D_n + \bar{\ell}_Y (Y_{n + 1}), $$
where $\bar{\ell}_Y (y) = \log \frac{\bar{g}_Y (y)}{f_Y (y)}$, until either time
$$ \sigma = \sigma^Y (A) = \min \{ n \geq 1: D_n \notin [0, A) \} $$
or $N_Y$ number of observations of experiment $Y$ are used. If $D_{\sigma} \geq A$ or $N_Y$ has been reached, we stop and declare the change.

\item If $D_{\sigma} < 0$, we start using observations $\{ X_n \}$ from the lower-quality experiment.

    \begin{enumerate}
    \item We first scale the undershoot $U_Y = D_{\sigma}$ by a factor $a_Y$.
    
    \item We update the statistic $\{ D_n \}$ starting at time $\sigma$ using
    $$ D_{n + 1} = D_n + \bar{\ell}_X (X_{n + 1}), $$
    where $\bar{\ell}_X (X_{n + 1}) = \log \frac{\bar{g}_X (x)}{f_X (x)}$, until either the statistic falls below $a_Y U_Y$ at time
    $$ \sigma^X (a_Y U_Y) = \min \{ n \geq \sigma + 1: D_n \notin [a_Y U_Y, 0) \} $$
    or $N_X$ number of observations of experiment $X$ are used. If $D_{\sigma^X (a_Y U_Y)} > 0$, we reset the statistic $D_n$ to $0$ and go to Step $1$.
    
    \item If $D_{\sigma^X (a_Y U_Y)} < a_Y U_Y$, we do not perform any experiment, but still scale the undershoot $U_X~=~D_{\sigma^X (a_Y U_Y)}$ by a factor of $a_X$ and use $a_X U_X$ as the starting level for deterministically updating $D_n$ as
    $$ D_{n + 1} = D_n + \mu. $$
    This continues until either $D_n$ goes above $a_Y U_Y$ or $N_0$ number of deterministic updates are used. At this time, we reset $D_n$ to $a_Y U_Y$ and continue updating $D_n$ as in Step $\textup{c}$.
    \end{enumerate}

\item Stop at the stopping time
\begin{align*}
\tau_{\scriptscriptstyle \textup{T2RC}}
& = \tau_{\scriptscriptstyle \textup{T2RC}} (A, a_X, a_Y, N_0, N_X, N_Y, \mu) \\
& = \inf \{ n \geq 1: D_n \geq A, \; \text{or $N_Y$ observations used for $Y$} \}.
\end{align*}
\end{enumerate}    
\end{algorithm}

\subsection{The \texorpdfstring{$m$}{}E-RCUSUM Algorithm}
\label{sec:mECUSUMAlgo}
Using the 2E-RCUSUM and the truncated 2E-RCUSUM algorithms, we can define the 3E-RCUSUM algorithm (with experiment $Z$ having the highest quality and experiment $X$ the lowest) as follows. We start by using experiment $Z$ while the statistic is above 0. Once it falls below 0, we scale the undershoot and use this scaled value as the starting point of the truncated 2E-RCUSUM algorithm for experiments $Y$ and $X$. Once the statistic crosses 0 from below or $N_Y$ observations of experiment $Y$ are used, we use experiment $Z$ again.

Similarly, we can define the 4E-RCUSUM algorithm as follows: the highest-quality experiment is performed while the statistic is above 0. When the statistic falls below 0, the scaled undershoot is used as the starting point of the truncated 3E-RCUSUM algorithm. Once the statistic goes above 0 or the maximum number of observations of the second-highest-quality experiment is used, the highest-quality experiment is used again.

Following this recursive description, we propose the $m$E-RCUSUM algorithm (Algorithm \ref{alg:me-rcusum} below). This algorithm starts by using the highest-quality experiment (experiment $m$) while the statistic is above 0. Once it falls below 0, we scale the undershoot and use this value as the starting point of the truncated $(m - 1)$E-RCUSUM algorithm for the remaining $m - 1$ experiments. Once the statistic crosses 0 from below or $N_{m - 1}$ observations of experiment $m - 1$ are used, experiment $m$ is performed again. 

For the description below, we revert back to the observation notation that we used in the problem formulation. Thus, for experiment $i$, observations are $\{X_n^i\}$, and the pre- and post-change densities are $f_i$ and $g_i$, respectively. Finally, the LFL for experiment $i$ is denoted by $\bar g_i$. 

\begin{algorithm}[$m$E-RCUSUM: $\Psi_{m\text{RC}} (A, a_1, \dots, a_m, \allowbreak N_0, \dots, N_{m - 1}, \mu) = \Psi_{m\text{RC}}$]
\label{alg:me-rcusum}

Fix $A > 0$, $a_i > 0$ ($i = 1, 2, \dots, m$), $N_j > 0$ ($j = 0, 1, \dots, m - 1$), and $\mu > 0$. The \textup{$m$E-RCUSUM} statistic $\{ D_n \}$ is defined as follows:

\begin{enumerate}
\item We start with $D_0 = 0$ and update $\{ D_n \}$ using the observations $\{ X_n^m \}$ from experiment $m$, the experiment with highest LFL quality, using
$$ D_{n + 1}  = D_n + \bar{\ell}_m (X_{n + 1}^m), $$
where $\bar{\ell}_m (x^m) = \log \frac{\bar{g}_m (x^m)}{f_m (x^m)}$, until 
$$ \sigma = \min \{ n \geq 1: D_n \notin [0, A) \}. $$
If $D_{\sigma} \geq A$, we stop and declare change.

\item If $D_{\sigma} < 0$, we start using observations $\{ X_n^{m - 1} \}$, $\{ X_n^{m - 2} \}$, \dots, $\{ X_n^1 \}$ from the lower-quality experiments and execute the truncated \textup{$(m - 1)$E-RCUSUM} algorithm for experiments $m - 1$, $m - 2$, up to $1$ as follows:

    \begin{enumerate}
    \item We first scale the undershoot $U_m = D_{\sigma}$ by a factor $a_m$, and use $a_m U_m$ as the new zero level.

    \item With $a_m U_m$ as the new zero level, the stopping threshold is set at $0$.

    \item We update the statistic starting at time $\sigma$ using the truncated \textup{$(m - 1)$E-RCUSUM} algorithm $\Psi_{\textup{T}(m - 1)\textup{RC}} (a_m U_m, 0, a_1, \dots, a_{m - 1}, N_0, \dots, \allowbreak N_{m - 1}, \mu)$.

    \item When $\Psi_{\textup{T}(m - 1)\textup{RC}} (a_m U_m, 0, a_1, \dots, a_{m - 1}, N_0,  \allowbreak \dots, N_{m - 1}, \mu)$ terminates, either because $D_n$ crosses $0$ from below or because $N_{m - 1}$ number of observations of experiment $m - 1$ have been used, we reset the statistic $D_n$ to $0$ and go to Step $1$.
    \end{enumerate}

\item Stop at the stopping time
\begin{align*}
\tau_{\scriptscriptstyle m\textup{RC}}
& = \tau_{\scriptscriptstyle m\textup{RC}} (A, a_1, \dots, a_m, N_0, \dots, N_{m - 1}, \mu) \\
& = \inf \{ n \geq 1: D_n \geq A \}.
\end{align*}
\end{enumerate}
\end{algorithm}

\begin{remark}
When the LFL of all experiments is replaced by their actual post-change distribution, the $m$E-RCUSUM algorithm specializes to the multiple-experiment CUSUM algorithm developed in \cite{lubenia2026}.
\end{remark}

% \begin{remark}
% When there is only one experiment, the $m$E-RCUSUM algorithm is similar to the robust data-efficient CUSUM (RDE-CUSUM) algorithm developed in \cite{hou2025}. The mechanism of $\Psi_{\text{T}1\text{RC} (A, 1, \infty, \infty, \mu)}$ is similar to that of the RDE-CUSUM algorithm, except that the undershoot of the latter is truncated at $-h$.
% \end{remark}

% ------- %
% Example %
% ------- %

\subsection{Example and Evolution of the \texorpdfstring{$m$}{}E-CUSUM Algorithm}
\label{sec:example}

To appreciate the structure of the $m$E-RCUSUM algorithm, we provide an illustration for the special case of three experiments. 
\begin{equation}
    \begin{split}
        \textup{Experiment 1}:&  f_1= \mathcal{N}(0, 1), \quad g_1=\mathcal{N}(\mu_1, 1), \quad \mu_1 \textup{ is unknown}. \\
        \textup{Experiment 2}: & f_2= \mathcal{N}(0, 1), \quad g_2=\mathcal{N}(\mu_2, 1), \quad \mu_2 \textup{ is unknown}. \\
         \textup{Experiment 3}: & f_3= \mathcal{N}(0, 1), \quad g_3=\mathcal{N}(\mu_3, 1), \quad \mu_3 \textup{ is unknown}. \\
    \end{split}
\end{equation}

Suppose the true values of the means be 
\begin{equation}
    \mu_1=2.5, \quad \mu_2 = 2.25, \quad \textup{ and } \quad \mu_3=2.
\end{equation}
In this case, we have 
\begin{equation}
\label{eq:org_KLOrder}
    D(g_1 \, \| \, f_1 ) > D(g_2 \, \| \, f_2 ) > D(g_3 \, \| \, f_3 ). 
\end{equation}
If these mean values were known, we could use the multiple experiment algorithm from \cite{lubenia2026}, with Experiment 1 as the best experiment and Experiment 2 as the next best. 

%Suppose the pre-change distribution of experiments 1, 2, and 3 is standard normal, i.e., $f_1 = f_2 = f_3 = \mathcal{N}(0, 1)$ with unknown post-change means $\mu_1$, $\mu_2$, and $\mu_3$. 
%Let the point estimates of the post-change means be $\hat{\mu}_1 = 2.5$, $\hat{\mu}_2 = 2.25$, and $\hat{\mu}_3 = 2$. In this case, experiment 1 has the highest quality while experiment 3 has the lowest quality.

Since the mean parameters $\{\mu_i\}$ are actually unknown, we can estimate them with confidence intervals. 
Suppose the following confidence intervals are associated with each point estimate: 
\begin{equation}
    \hat{\mu}_1 \in [0.5, 4.5], \quad
    \hat{\mu}_2 \in [0.75, 3.75], \quad \hat{\mu}_3 \in [1, 3].
\end{equation}
If we choose to design a robust test for this setup, we would use the LFLs, which in this case are
\begin{equation}
    \bar g_1 = \mathcal{N}(0.5, 1), \quad
    \bar g_2 = \mathcal{N}(0.75, 1), \quad
    \bar g_3 = \mathcal{N}(1, 1).
\end{equation}
This gives the order
\begin{equation}
\label{eq:new_KLOrder}
    D(\bar g_1 \, \| \, f_1 ) < D(\bar g_2 \, \| \, f_2 ) < D(\bar g_3 \, \| \, f_3 ),
\end{equation}
which is the opposite of the order given in \eqref{eq:org_KLOrder} based on the true values of the parameters. 
When employing the $m$E-RCUSUM algorithm, we would use Experiment 3 as our best experiment and Experiment 2 as the next best. 

\begin{figure}[!t]
\centering
\includegraphics[width = 0.65\linewidth]{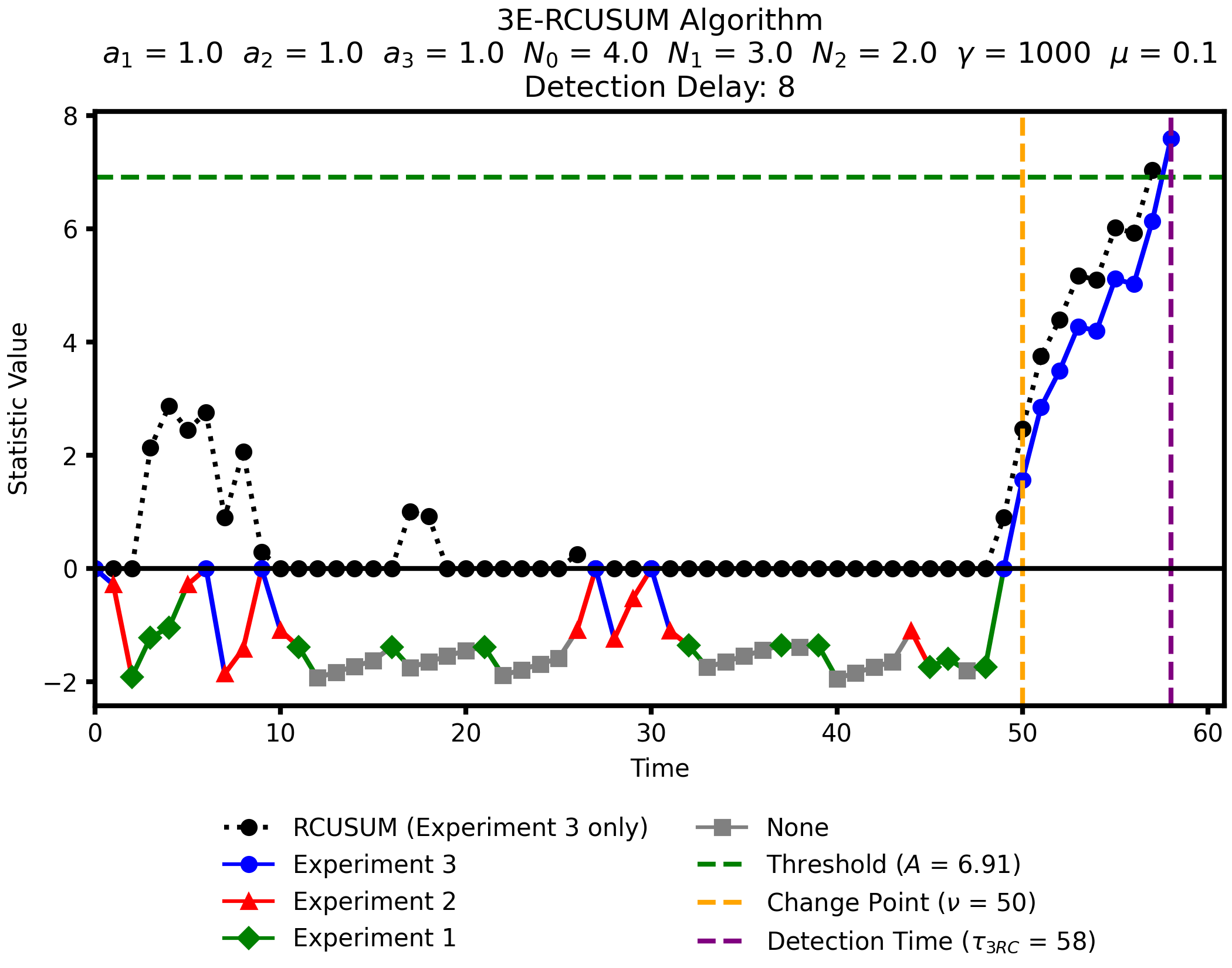}
\caption{An evolution of the statistic $D_n$ of the 3E-RCUSUM algorithm for $f_1 = f_2 = f_3 = \mathcal{N}(0, 1)$, $\bar{g}_1 = \mathcal{N}(0.5, 1)$, $\bar{g}_2 = \mathcal{N}(0.75, 1)$, $\bar{g}_3 = \mathcal{N}(1, 1)$, $\nu = 50$, and $\gamma = 1\,000$, with $a_1 = a_2 = a_3 = 1.0$, $N_2 = 2.0$, $N_1 = 3.0$, $N_0 = 4.0$, and $\mu = 0.1$. The parameter values correspond to POR values of $\text{POR}_3 = 0.15$, $\text{POR}_2 = 0.15$, $\text{POR}_1 = 0.25$, and $\text{POR}_0 = 0.45$. The decision-maker performs experiment 3 whenever $D_n \geq 0$. When $D_n < 0$, the observer performs experiments 2 and 1 within their thresholds, or does not perform any experiment after experiment 1. Also shown is the RCUSUM algorithm for experiment 3. The performance of the 3E-RCUSUM algorithm is comparable to that of the best experiment.}
\label{fig:3e-rcusum}
\end{figure}

Fig.~\ref{fig:3e-rcusum} illustrates an evolution of the 3E-RCUSUM algorithm for the above example. 
In our algorithm, we set the parameters as follows:
\begin{align*}
\textup{Scaling parameters: }& \quad a_1 = a_2 = a_3 = 1.0, \\
\textup{Truncation parameters: }& \quad  N_2 = 2.0, \; N_1 = 3.0, \;  N_0 = 4.0, \ \text{and} \\
\textup{Data-efficiency parameter: } &\quad  \mu = 0.1.
\end{align*}
These choices of parameters give the following POR values:
$$ \text{POR}_3 = 0.15, \;\text{POR}_2 = 0.15, \;\text{POR}_1 = 0.25, \ \;\text{and} \ \;\text{POR}_0 = 0.45. $$
Thus, before the change-point, experiments 2 and 3 are each used only 15\% of the time, experiment 1 only 25\% of the time, and no experiment is used 45\% of the time. With $\gamma = 1\,000$, the threshold was set to $A = \log \gamma$, a choice justified in Theorem~\ref{thm:ALRFA}.

Fig.~\ref{fig:3e-rcusum} also presents the evolution of the RCUSUM algorithm for the experiment 3. The RCUSUM algorithm was first studied in \cite{unnikrishnan2011} for the case of single experiment with unknown post-change model. 
Since the RCUSUM is designed without any concern or constraint on the cost of sampling or experiments, the experiment it uses, the experiment 3, is observed $49$ times before the change-point of $50$, compared to only six observations for the 3E-RCUSUM algorithm. Despite this, the detection delay of the two algorithms are comparable. The behavior shown in Fig.~\ref{fig:3e-rcusum} is typical of the 3E-RCUSUM algorithm. We formalize this observation in the theorems in the next section, where it is shown that the 3E-RCUSUM algorithm is asymptotically optimal. Specifically, he 3E-RCUSUM algorithm achieves the performance of the RCUSUM algorithm executed using the best experiment 3, as the rate of false alarms goes to zero.

% ---------------------------------------------- %
% Robust Optimality of the $m$E-RCUSUM Algorithm %
% ---------------------------------------------- %

\section{Robust Optimality of the \texorpdfstring{$m$}{}E-RCUSUM Algorithm}
\label{sec:optimality}

In this section, we prove the robust optimality of the $m$E-RCUSUM algorithm. We first establish a universal lower bound on any policy $\Psi$ which solves the robust multiple-experiment QCD problem described in \eqref{eq:prob-formulation}.

\begin{theorem}
\label{thm:lower-bound}
Suppose Assumption~1 and Assumption~2 are satisfied and 
$$ 0 < D(\bar{g}_1 \ \Vert \ f_1) \leq \dots \leq D(\bar{g}_m \ \Vert \ f_m) < \infty. $$
Let $\Psi$ be a policy satisfying $\mathsf{ARLFA}(\Psi) \geq \gamma$ and $\mathsf{POR}_i (\Psi) \leq \beta_i$ for $i = 1, 2, \dots, m$, for fixed $\gamma > 0$ and $0 < \beta_i < 1$ where $\ds \sum_{i = 1}^m \beta_i < 1$. Then as $\gamma \to \infty$,
$$ \sup_{G \in \mathcal{G}} \ \mathsf{WADD}_G (\Psi) \geq \frac{\log \gamma}{D(\bar{g}_m \ \Vert \ f_m)} (1 + o(1)). $$
\end{theorem}

\begin{IEEEproof}
By Assumption~1, an LFL $\bar{g}_i \in \mathcal{G}_i$ exists for each experiment $i = 1, \dots, m$. Consider $\bar{G} = (\bar{g}_1, \dots, \bar{g}_m) \in \mathcal{G}$. Thus, for any policy $\Psi$,
$$ \sup_{G \in \mathcal{G}} \ \mathsf{WADD}_G (\Psi) \geq \mathsf{WADD}_{\bar{G}}   (\Psi) \geq \inf_\Psi \ \mathsf{WADD}_{\bar{G}} (\Psi), $$
where the infimum is taken over all policies satisfying the ARLFA and POR constraints. 
By Corollary~1 in \cite{lubenia2026}, as $\gamma \to \infty$,
$$ \inf_\Psi \ \mathsf{WADD}_{\bar{G}} (\Psi) \sim \frac{\log \gamma}{D(\bar{g}_m \ \Vert \ f_m)}.$$
Thus, as $\gamma \to \infty$,
$$ \sup_{G \in \mathcal{G}} \ \mathsf{WADD}_G (\Psi) \geq \frac{\log \gamma}{D(\bar{g}_m \ \Vert \ f_m)} (1 + o(1)). $$
\end{IEEEproof}

\medskip
Our goal in the rest of the section is to show that the $m$E-RCUSUM algorithm $\Psi_{m\text{RC}}$ achieves the above lower bound. 
To this end, we first provide results regarding its ARLFA and POR performances, i.e., we provide guidelines on how to design the algorithm to satisfy a given set of ARLFA and POR constraints. In the theorem below, we define $\tau_{\scriptscriptstyle \text{RC}}^i (A)$ as the stopping time of the RCUSUM test using experiment $i$ alone, with stopping threshold $A$.

\begin{theorem}
\label{thm:ALRFA}
Consider the \textup{$m$E-RCUSUM} policy $\Psi_{m\textup{RC}}$ with
$$ 0 < D(\bar{g}_1 \ \Vert \ f_1) \leq \dots \leq D(\bar{g}_m \ \Vert \ f_m) < \infty. $$
For any fixed $a_i > 0$ ($i = 1, \dots, m$), $N_j > 0$ ($j = 0, \dots, \allowbreak m - 1$), and $\mu > 0$, setting $A = \log \gamma$ guarantees
$$ \mathsf{ARLFA}(\Psi_{m\textup{RC}}) \geq \gamma. $$
\end{theorem}

\begin{IEEEproof}
The proof is similar to the proof of Theorem 1 in \cite{lubenia2026}. We provide the complete proof here for convenience. 

Since the RCUSUM algorithm for any experiment $i$ is the classical CUSUM algorithm with pre-change density $f_i$ and post-change density $\bar{g}_i$, we have the following well-known result for the lower bound of the RCUSUM statistic based on \cite{lorden1971}:
$$ \mathsf{E}_\infty [\tau_{\scriptscriptstyle \text{RC}}^i (A)] \geq e^A. $$
Each time the statistic $D_n$ cross zero from below, it is reset to zero. At this time, the process renews itself and the process from that point on has the same distribution as the process that started at time $1$. Since the obervations are independent, this implies that in the $m$E-RCUSUM algorithm, the sojourns of the statistic $D_n$ below zero (when lower-quality experiments are used or when no experiment is performed) only delay the statistic from reaching the threshold $A$. In fact, if we remove all the sojourns of the statistic $D_n$ below zero from 
the $m$E-RCUSUM algorithm, the algorithm will be statistically identical to the RCUSUM algrorithm for experiment $m$ (the one with the highest LFL information quality). 

Thus, conditioning on the event that the highest-quality experiment $m$ is used at every observation time decreases the expected stopping time:
\begin{align*}
\mathsf{ARLFA}(\Psi_{m\text{RC}})
& = \mathsf{E}_\infty [\tau_{\scriptscriptstyle m\text{RC}}] \\
& \geq \mathsf{E}_\infty [\tau_{\scriptscriptstyle m\text{RC}} \mid S_n^m = 1 \ \forall n] \\
& = \mathsf{E}_\infty [\tau_{\scriptscriptstyle \text{RC}}^m (A)] \geq e^A.
\end{align*}
Setting $A = \log \gamma$, we get
$$ \text{ARLFA}(\Psi_{m\text{RC}}) \geq e^{\log \gamma} = \gamma. $$
\end{IEEEproof}

Theorem~\ref{thm:ALRFA} shows that the ARLFA constraint is satisfied by setting a threshold $A$ similar to the one used in the classical CUSUM algorithm. We now prove a theorem on satisfying the POR constraints.

\medskip
\begin{theorem}
\label{thm:POR}

Consider the \textup{$m$E-RCUSUM} policy $\Psi_{m\textup{RC}}$ with
$$ 0 < D(\bar{g}_1 \ \Vert \ f_1) \leq \dots \leq D(\bar{g}_m \ \Vert \ f_m) < \infty. $$
For any fixed threshold $A$, there exist values of $a_i > 0$ ($i = 1, \dots, m$), $N_j > 0$ ($j = 0, \dots, m - 1$), and $\mu > 0$, all independent of $A$, such that for each $i = 1, \dots, m$,
$$ \mathsf{POR}_i (\Psi_{m\textup{RC}}) \leq \beta_i $$
for any given constraint $\beta_i \in (0, 1)$ with $\ds \sum_{i = 1}^m \beta_i < 1$.
\end{theorem}

\begin{IEEEproof}
The statistic update defining the $m$E-RCUSUM policy $\Psi_{m\textup{RC}}$ is similar to that of the DE$m$E-CUSUM policy $\Psi_{\text{DE}m\text{E}}$ (Algorithm~9 in \cite{lubenia2026}) where the former uses the pre-change density $f_i$ and post-change density $\bar{g}_i$, while the latter uses $f_i$ and $g_i$, respectively, for each $i = 1, \dots, m$. Moreover, the structure of $\Psi_{m\textup{RC}}$ and $\Psi_{\text{DE}m\text{E}}$ is identical, except for the log-likelihood ratio used. Thus, the claim for $\Psi_{m\textup{RC}}$ follows immediately from Theorem~5(2) in \cite{lubenia2026} for $\Psi_{\text{DE}m\text{E}}$.
\end{IEEEproof}

\medskip

Theorem~\ref{thm:POR} shows that the POR constraints are satisfied by controlling the scaling factors and truncation parameters, independent of the threshold $A$. We note that when a truncation parameter is not integer-valued, then that truncation level is achieved by randomization at the nearest integer levels. 

\medskip
We now provide an upper bound for the delay of the $m$E-RCUSUM algorithm.

\begin{theorem}
\label{thm:upper-bound}

Suppose Assumptions 1 and 2 hold. Then for any fixed choices of $a_i > 0$ ($i = 1, \dots, m$), $N_j > 0$ \linebreak ($j = 0, \dots, m - 1$), and $\mu > 0$, as $A \to \infty$,
$$ \sup_{G \in \mathcal{G}} \ \mathsf{WADD}_G (\Psi_{m\textup{RC}}) \leq \frac{A}{D(\bar{g}_m \ \Vert \ f_m)}(1 + o(1)). $$
\end{theorem}

\begin{IEEEproof}
Since the $m$E-RCUSUM algorithm has the same structure as the DE$m$E-CUSUM algorithm, by Theorem~5(3) in \cite{lubenia2026}, when the statistic $D_n$ of the $m$E-RCUSUM algorithm falls below 0, due to the truncation parameters, the truncated $(m - 1)$E-RCUSUM algorithm resets $D_n$ to 0 after using at most
$$ K_{m - 1} = \sum_{r = 0}^{m - 1} \prod_{k = (m - 1) - r}^{m - 1} N_k $$
observations. This bound is independent of the threshold $A$, the time the statistic falls below 0, and the environment $G$.

We now obtain a bound on
$$ \mathsf{WADD}_G (\Psi_{\scriptscriptstyle m\text{RC}}) = \sup_{\nu \geq 1} \esssup \mathsf{E}_\nu^G [(\tau_{\scriptscriptstyle m\text{RC}} - \nu + 1)^+ \mid I_{\nu - 1}] $$
for any fixed $G \in \mathcal{G}$.

We note that for a fixed $\nu$, the conditional expectation depends on $I_{\nu - 1}$ through the $m$E-RCUSUM statistic $D_{\nu - 1}$ only. As a result we have
\begin{equation}
\mathsf{E}_\nu^G [(\tau_{\scriptscriptstyle m\text{RC}} - \nu + 1)^+ \mid I_{\nu - 1}] \leq \mathsf{E}_1^G [\tau_{\scriptscriptstyle m\text{RC}}] + K_{m - 1}. \label{eq:ave-delay-bound}
\end{equation}
This is because the worst possible choice of the change point $\nu$ and the worst-possible information vector $I_{\nu-1}$ will make sure that when the change occurs at a certain time $\nu$, the algorithm is in the data-efficiency mode, i.e., performing no experiments. From this state, the time taken by the statistic $D_n$ to reach $0$ from below is bounded by $K_{m-1}$. Once the statistic reached $0$ from below, it is reset to $0$. From this point onward, the average detection delay will be $\mathsf{E}_1^G [\tau_{\scriptscriptstyle m\text{RC}}]$. 
Since the right-hand side of \eqref{eq:ave-delay-bound} does not depend on $I_{\nu - 1}$ and $\nu$, taking the supremum over $\nu$ of the essential supremum of both sides gives us
\begin{equation}
\mathsf{WADD}_G (\Psi_{m\text{RC}}) \leq \mathsf{E}_1^G [\tau_{\scriptscriptstyle m\text{RC}}] + K_{m - 1}. \label{eq:wadd-mrc-bound}
\end{equation}

We now determine a bound for $\mathsf{E}_1^G [\tau_{\scriptscriptstyle m\text{RC}}]$.

Let $\tau_{\scriptscriptstyle \text{RC}}^m$ be the stopping time of the RCUSUM algorithm in \cite{unnikrishnan2011} where experiment $m$ is used. Let the statistic of this algorithm be
\begin{align*}
& W_{n + 1} = \left(W_n + \log \frac{\bar{g}_m (X_{n + 1}^m)}{f_m (X_{n + 1}^m)}\right)^+,  \; \; \text{with} \ W_0 = 0, \\
& \tau_{\scriptscriptstyle \text{RC}}^m = \inf \{ n \geq 1: W_n \geq A \}.
\end{align*}
The evolution of the $m$E-RCUSUM statistic $D_n$ and the RCUSUM statistic $W_n$ are statistically identical, except for the sojourns of $D_n$ below zero.  To clarify this, we now define 
$$ \lambda_A = \inf \{ n \geq 1: D_n \notin [0, A) \} = \inf \{ n \geq 1: W_n \notin [0, A) \}, $$
the first time the $m$E-RCUSUM statistic $D_n$ (or the RCUSUM statistic $W_n$) either falls below 0 or crosses $A$ from below. Thus, if change occurs at time $\nu=1$, the starting at time $1$, the algorithms $\tau_{\scriptscriptstyle \text{RC}}^m $ and $\Psi_{m\text{RC}}$ are identical until either the statistics cross $A$ or go below zero. If $D_n$ goes below zero, as explained at the start of the proof, the number of observations made is bounded by $K_{m - 1}$.  In the algorithm $\tau_{\scriptscriptstyle \text{RC}}^m$, when the statistic $D_n$ crosses zero from below, the statistic is reset to zero. This starts a new renewal cycle. Consequently, the evolution of the statistic $D_n$ is statistically identical to that of the statistic $W_n$ starting at time $\lambda_A$ (conditioned on $W_{\lambda_A} < 0$). 

The number of such renewal cycles is a geometrically distributed random variable with probability $\mathsf{P}_1^G (D_{\lambda_A} \geq A)$ of successfully crossing $A$ from below. Thus, the average number of cycles is the mean of the said geometric random variable: 
$$
\frac{1}{\mathsf{P}_1^G (D_{\lambda_A} \geq A)}. 
$$
Applying the reasoning above to \eqref{eq:wadd-mrc-bound} gives us
\begin{align*}
\mathsf{WADD}_G (\Psi_{m\text{RC}})
& \leq \mathsf{E}_1^G [\tau_{\scriptscriptstyle m\text{RC}}] + K_{m - 1} \\
& \leq \mathsf{E}_1^{G_m} [\tau_{\scriptscriptstyle \text{RC}}^m] + \frac{K_{m - 1}}{\mathsf{P}_1^G (D_{\lambda_A} \geq A)} + K_{m - 1}.
\end{align*}
Here the term $ \frac{K_{m - 1}}{\mathsf{P}_1^G (D_{\lambda_A} \geq A)}$ appears as a bound on the time spent by the statistic $D_n$ below the level zero. 
Taking the supremum over $G \in \mathcal{G}$ of both sides, we get
\begin{align}
\sup_{G \in \mathcal{G}} \mathsf{WADD}_G (\Psi_{m\text{RC}}) \notag 
&\leq \sup_{G \in \mathcal{G}} \mathsf{E}_1^{G_m} [\tau_{\scriptscriptstyle \text{RC}}^m] + \frac{K_{m - 1}}{\ds \inf_{G \in \mathcal{G}} \mathsf{P}_1^G (D_{\lambda_A} \geq A)} + K_{m - 1}\\
 &= \sup_{G_m \in \mathcal{G}_m} \mathsf{E}_1^{G_m} [\tau_{\scriptscriptstyle \text{RC}}^m] + \frac{K_{m - 1}}{\ds \inf_{G \in \mathcal{G}} \mathsf{P}_1^G (D_{\lambda_A} \geq A)} + K_{m - 1}. \label{eq:sup-wadd}
\end{align}
Note that the supremum in the first term on the right-hand side of the quality is over the post-change law $G_m~\in~\mathcal{G}_m$ for experiment $m$ since the expectation is of the stopping time $\tau_{\scriptscriptstyle \text{RC}}^m$ of the RCUSUM algorithm. 

Now, due to Assumptions~1 and 2, and the proof of Theorem III.2 in \cite{unnikrishnan2011}, we have
\begin{equation}
\sup_{G_m \in \mathcal{G}_m} \ \mathsf{E}_1^{G_m} [\tau_{\scriptscriptstyle \text{RC}}^m] = \mathsf{E}_1^{\bar{G}_m} [\tau_{\scriptscriptstyle \text{RC}}^m]. \label{eq:sup-g}
\end{equation}
Also, by the same set of assumptions, and following the arguments in the proof of Theorem IV.4 in \cite{hou2025}, we have
\begin{align}
& \mathsf{P}_1^G (D_{\lambda_A} \geq A) \notag \\
& \quad \geq \mathsf{P}_1^G (D_n \geq 0 \ \text{for all} \ n) \notag \\
& \quad = \lim_{N \to \infty} \ \mathsf{P}_1^G (D_n \geq 0 \ \text{for all} \ n \leq N) \notag \\
& \quad = \lim_{N \to \infty} \ \mathsf{P}_1^{G_m} \left( \max_{1 \leq n \leq N} \max_{1 \leq k \leq n} \sum_{i = k}^n \bar{\ell}_m (X_i^m) \geq 0 \right) \notag \\
& \quad \geq \lim_{N \to \infty} \ \mathsf{P}_1^{\bar{G}_m} \left( \max_{1 \leq n \leq N} \max_{1 \leq k \leq n} \sum_{i = k}^n \bar{\ell}_m (X_i^m) \geq 0 \right) \notag \\
& \quad = \mathsf{P}_1^{\bar{G}_m} \left( \max_{1 \leq k \leq n} \sum_{i = k}^n \bar{\ell}_m (X_i^m) \geq 0 \ \text{for all} \ n \right). \label{eq:p-geq-q}
\end{align}
The probability on the extreme right, 
$$ q := \mathsf{P}_1^{\bar{G}_m} \left( \max_{1 \leq k \leq n} \sum_{i = k}^n \bar{\ell}_m (X_i^m) \geq 0 \ \text{for all} \ n \right) > 0. $$
This follows from the results of renewal theory \cite{woodroofe1982}. We note that the probability $q$ depends on $\bar{G}_m$ and not on any other $G \in \mathcal{G}$. Also, $q$ does not depend on threshold $A$.

Taking the infimum over $G \in \mathcal{G}$ of both sides of \eqref{eq:p-geq-q}, we get
\begin{equation}
\inf_{G \in \mathcal{G}} \ \mathsf{P}_1^G (D_{\lambda_A} \geq A) \geq q > 0. \label{eq:inf-p}
\end{equation}

Substituting \eqref{eq:sup-g} and \eqref{eq:inf-p} into \eqref{eq:sup-wadd}, we get,
\begin{align*}
 \sup_{G \in \mathcal{G}} \ \mathsf{WADD}_G (\Psi_{m\text{RC}}) & \quad \leq \sup_{G_m \in \mathcal{G}_m} \ \mathsf{E}_1^{G_m} [\tau_{\scriptscriptstyle \text{RC}}^m] + \frac{K_{m - 1}}{\ds \inf_{G \in \mathcal{G}} \ \mathsf{P}_1^G (D_{\lambda_A} \geq A)} + K_{m - 1} \\
& \quad \leq \mathsf{E}_1^{\bar{G}_m} [\tau_{\scriptscriptstyle \text{RC}}^m] + \frac{K_{m - 1}}{q} + K_{m - 1} \\
& \quad \leq \frac{A}{D(\bar{g}_m \ \Vert \ f_m)}(1 + o(1)) + \frac{K_{m - 1}}{q} + K_{m - 1}, \quad \textup{ as $A \to \infty$}, \\
& \quad = \frac{A}{D(\bar{g}_m \ \Vert \ f_m)}(1 + o(1)), \quad \textup{ as $A \to \infty$}.
\end{align*}
Here, the last inequality follows from the analysis of the CUSUM algorithm in \cite{lorden1971}, and the last equality follows since $\frac{K_{m - 1}}{q} + K_{m - 1}$ is a constant which does not depend on $A$.
\end{IEEEproof}

\medskip

Theorem~\ref{thm:upper-bound} shows that the delay of the $m$E-RCUSUM algorithm is within a constant of the delay of the CUSUM algorithm for the highest-quality experiment $m$.

Finally, based on the results above, we now state and prove the result on the asymptotic robust optimality of the $m$E-RCUSUM algorithm.

\begin{theorem}
\label{thm:robust-optimality}

If Assumptions~1 and 2 hold, then the \textup{$m$E-RCUSUM} algorithm is asymptotically robust optimal for each fixed $\beta_i$ ($i = 1, \dots, m$) with $\ds \sum_{i = 1}^m \beta_i < 1$, as $\gamma \to \infty$.
\end{theorem}

\begin{IEEEproof}
By Theorem~\ref{thm:ALRFA}, setting $A = \log \gamma$ guarantees that $\mathsf{ARLFA}(\Psi_{m\text{RC}}) \geq \gamma$. By Theorem~\ref{thm:POR}, we can set the values of $a_i > 0$ ($i = 1, \dots, m$), $N_j > 0$ ($j = 0, \dots, m - 1$), and $\mu > 0$, all independent of $A$, to satisfy $\mathsf{POR}_i (\Psi_{m\textup{RC}}) \leq \beta_i$ for each $i = 1, \dots, m$. With these parameter values, Theorem~\ref{thm:upper-bound} gives us, as $A \to \infty$ (equivalently, $\gamma \to \infty$),
\begin{align*}
\sup_{G \in \mathcal{G}} \ \mathsf{WADD}_G (\Psi_{m\textup{RC}})
& \leq \frac{A}{D(\bar{g}_m \ \Vert \ f_m)}(1 + o(1)) \\
& = \frac{\log \gamma}{D(\bar{g}_m \ \Vert \ f_m)}(1 + o(1)).
\end{align*}
Now, by Theorem~\ref{thm:lower-bound}, as $\gamma \to \infty$,
$$ \sup_{G \in \mathcal{G}} \ \mathsf{WADD}_G (\Psi_{m\textup{RC}}) \geq \frac{\log \gamma}{D(\bar{g}_m \ \Vert \ f_m)}(1 + o(1)). $$
Therefore, as $\gamma \to \infty$,
$$ \sup_{G \in \mathcal{G}} \ \mathsf{WADD}_G (\Psi_{m\textup{RC}}) = \frac{\log \gamma}{D(\bar{g}_m \ \Vert \ f_m)}(1 + o(1)), $$
i.e., the $m$E-RCUSUM algorithm is first-order asymptotically optimal.
\end{IEEEproof}

% ---------- %
% Conclusion %
% ---------- %

\section{Conclusion}
\label{sec:conclusion}

We developed the multiple-experiment robust cumulative sum ($m$E-RCUSUM) algorithm which utilizes multiple experiments to quickly detect a change in a stochastic process where the post-change distribution is unknown but the decision-maker believes the distribution belongs to a family of distributions. The algorithm uses the least favorable law (LFL) for the post-change model of each experiment to detect the change and design the experiments. 
We highlighted four facts regarding the $m$E-RCUSUM algorithm. First, the average run length to a false alarm (ARLFA) constraint is satisfied by setting a threshold $A$ similar to the one used in the CUSUM algorithm. Second, the pre-change observation ratio (POR) constraints are satisfied by controlling the scaling factors and truncation parameters, independent of the threshold $A$. Third, the delay of the $m$E-RCUSUM algorithm is within a constant of the delay of the CUSUM algorithm for experiment $m$. Finally, the algorithm is first-order asymptotically robust optimal.

The main insight provided by this work is as follows. Suppose a change detector has multiple experiments of differing quality to choose from and the exact post-change law is unknown. Consequently, the true order of the quality of experiments is also unknown. 
In this case,  to be robust optimal, the change detector must order the experiments according to the quality provided by the LFLs. The best delay performance achievable in this setup, for any given set of constraints on the ARLFA and POR values, is governed by the experiment with the highest LFL quality.

% if have a single appendix:
%\appendix[Proof of the Zonklar Equations]
% or
%\appendix  % for no appendix heading
% do not use \section anymore after \appendix, only \section*
% is possibly needed

% use appendices with more than one appendix
% then use \section to start each appendix
% you must declare a \section before using any
% \subsection or using \label (\appendices by itself
% starts a section numbered zero.)
%

% -------- %
% Appendix %
% -------- %

% \appendices
% \section{Proof of the First Zonklar Equation}
% Appendix one text goes here.

% you can choose not to have a title for an appendix
% if you want by leaving the argument blank
% \section{}
% Appendix two text goes here.

% -------------- %
% Acknowledgment %
% -------------- %

% use section* for acknowledgment
\section*{Acknowledgment}

This work was supported in part by the U.S. National Science Foundation under Grant 2427909.

% Can use something like this to put references on a page
% by themselves when using endfloat and the captionsoff option.
\ifCLASSOPTIONcaptionsoff
  \newpage
\fi

% ---------- %
% References %
% ---------- %

% trigger a \newpage just before the given reference
% number - used to balance the columns on the last page
% adjust value as needed - may need to be readjusted if
% the document is modified later
%\IEEEtriggeratref{8}
% The "triggered" command can be changed if desired:
%\IEEEtriggercmd{\enlargethispage{-5in}}

% references section

% can use a bibliography generated by BibTeX as a .bbl file
% BibTeX documentation can be easily obtained at:
% http://mirror.ctan.org/biblio/bibtex/contrib/doc/
% The IEEEtran BibTeX style support page is at:
% http://www.michaelshell.org/tex/ieeetran/bibtex/
%\bibliographystyle{IEEEtran}
% argument is your BibTeX string definitions and bibliography database(s)
%\bibliography{IEEEabrv,../bib/paper}
%
% <OR> manually copy in the resultant .bbl file
% set second argument of \begin to the number of references
% (used to reserve space for the reference number labels box)
% \begin{thebibliography}{1}

% \bibitem{IEEEhowto:kopka}
% H.~Kopka and P.~W. Daly, \emph{A Guide to \LaTeX}, 3rd~ed.\hskip 1em plus
%   0.5em minus 0.4em\relax Harlow, England: Addison-Wesley, 1999.

% \end{thebibliography}

\bibliographystyle{IEEEtran}
\bibliography{references}

\end{document}